\documentclass[10pt,twocolumn,twoside]{IEEEtran}

\IEEEoverridecommandlockouts                              	

\usepackage{amssymb,amsmath,amsfonts,amsthm}
\usepackage{algorithm,algorithmic}
\usepackage{cite}
\usepackage{float}
\usepackage{epsfig}
\usepackage{graphicx}
\usepackage{graphics}
\usepackage{subfigure}
\usepackage{url}
\usepackage{xspace}
\usepackage[usenames,dvipsnames]{color}
\usepackage{soul}
\usepackage{mathtools}

\floatstyle{ruled}
\newfloat{model}{H}{mod}
\floatname{model}{\footnotesize Model}
\newfloat{notatio}{H}{not}
\floatname{notatio}{\footnotesize Notation}

\newenvironment{varalgorithm}[1]
  {\algorithm\renewcommand{\thealgorithm}{#1}}
  {\endalgorithm}

\newenvironment{list4}{
	\begin{list}{$\bullet$}{%
			\setlength{\itemsep}{0.05cm}
			\setlength{\labelsep}{0.2cm}
			\setlength{\labelwidth}{0.3cm}
			\setlength{\parsep}{0in} 
			\setlength{\parskip}{0in}
			\setlength{\topsep}{0in} 
			\setlength{\partopsep}{0in}
			\setlength{\leftmargin}{0.16in}}}
	{\end{list}}

\newenvironment{list4a}{
	\begin{list}{$\bullet$}{%
			\setlength{\itemsep}{0.05cm}
			\setlength{\labelsep}{0.2cm}
			\setlength{\labelwidth}{0.3cm}
			\setlength{\parsep}{0in} 
			\setlength{\parskip}{0in}
			\setlength{\topsep}{0in} 
			\setlength{\partopsep}{0in}
			\setlength{\leftmargin}{0.16in}}}
	{\end{list}}

\newenvironment{list5}{
	\begin{list}{$\bullet$}{%
			\setlength{\itemsep}{0.05cm}
			\setlength{\labelsep}{0.2cm}
			\setlength{\labelwidth}{0.3cm}
			\setlength{\parsep}{0in} 
			\setlength{\parskip}{0in}
			\setlength{\topsep}{0in} 
			\setlength{\partopsep}{0in}
			\setlength{\leftmargin}{0.18in}}}
	{\end{list}}

\usepackage{url,changebar,bm,xspace,dsfont}
\let\mathbb=\mathds 
\def\day{\mathrm{day}} 

\newtheorem{theorem}{Theorem}

\newtheorem{defn}{Definition}

\newtheorem{assum}{Assumption}

\newtheorem{remark}{Remark}

\ifCLASSINFOpdf
 \else
\fi

\usepackage{amsmath,amsfonts,amssymb,amscd}
\usepackage{capt-of}
\usepackage{color}
\usepackage{cite}
\usepackage{verbatim}
\usepackage{hyperref}

\hypersetup{
    colorlinks = true,
    linkcolor = [rgb]{0,0,1},
    anchorcolor = [rgb]{0,0,1},
    citecolor = [rgb]{0.9,0.5,0},
    filecolor = [rgb]{0,0,1},
    urlcolor = [rgb]{0.9,0.5,0},
    bookmarksopen = true,
    bookmarksnumbered = true,
    breaklinks = true,
    linktocpage,
    colorlinks = true,
    linkcolor = [rgb]{0.2,0.6,0.2},
    urlcolor  = [rgb]{0,0,1},
    citecolor = [rgb]{0.9,0.5,0},
    anchorcolor = [rgb]{0.2,0.6,0.2},
}

\DeclareMathOperator*{\argmin}{arg\,min}

\begin{document}

\title{\LARGE \bf Distributed Finite Time $k$-means Clustering \\ with Quantized Communucation and Transmission Stopping}

\author{Apostolos~I.~Rikos, Gabriele Oliva, Christoforos N. Hadjicostis, and Karl~H.~Johansson
\thanks{Apostolos~I.~Rikos and K.~H.~Johansson are with the Division of Decision and Control Systems, KTH Royal Institute of Technology, SE-100 44 Stockholm, Sweden. E-mails: {\tt \{rikos,kallej\}@kth.se}.}
\thanks{Gabriele~Oliva is with the Unit of Automatic Control, Department of Engineering, Universit\`a Campus Bio-Medico di Roma, via \'Alvaro del Portillo 21, 00128, Rome, Italy. E-mail:{\tt~g.oliva@unicampus.it}.}
\thanks{C. N. Hadjicostis is with the Department of Electrical and Computer Engineering, University of Cyprus, 1678 Nicosia, Cyprus.  E-mail:{\tt~chadjic@ucy.ac.cy}.}
}

\maketitle
\thispagestyle{empty}
\pagestyle{empty}

%
%
%
%
\begin{abstract} 
In this paper, we present a distributed algorithm which implements the $k$-means algorithm in a distributed fashion for multi-agent systems with directed communication links. 
The goal of $k$-means is to partition the network's agents in mutually exclusive sets (groups) such that agents in the same set have (and possibly share) similar information and are able to calculate a representative value for their group. 
During the operation of our distributed algorithm, each node (i) transmits quantized values in an event-driven fashion, and (ii) exhibits distributed stopping capabilities. 
Transmitting quantized values leads to more efficient usage of the available bandwidth and reduces the communication bottleneck. 
Also, in order to preserve available resources, nodes are able to distributively determine whether they can terminate the operation of the proposed algorithm. 
We characterize the properties of the proposed distributed algorithm and show that its execution (on any static and strongly connected digraph) will partition all agents to mutually exclusive clusters in finite time. 
We conclude with examples that illustrate the operation, performance, and potential advantages of the proposed algorithm.
\end{abstract}


%
%
%
%
\section{Introduction}\label{intro}
Data Clustering is a fundamental problem whereby data is clustered in groups and a representative value is identified for each group. Such methods are adopted in a broad variety of different applications, ranging from  customer segmentation~\cite{li2022customer} to cybersecurity~\cite{jain2022k}.
Notably, in the case of wireless sensor networks, a large amount of data is typically generated or sensed \cite{ferjaoui2020data}.
In this view, the ability of a set of agents to collectively cluster their sensed data would allow to contain the overall amount of information and to establish functional connections among the agents, e.g., by identifying other agents with similar values.

In the literature, there have been various works at distributing clustering (e.g., see \cite{2000:Dhillon_Modha, 2006:Bandyopadhyay} and references therein) and recently on distributed algorithms such as the $k$-means~\cite{2015:Oliva, 2017:Zheng_Qin, 2008:Forero_Giannakis} and the C-means~\cite{faramondi2019distributed}. 
However, most clustering algorithms feature a message exchanges consisting of floating point values.
This leads to a significant increase in the computational and bandwidth requirements and may be responsible for introducing quantization errors or approximations.

\noindent
\textbf{Main Contributions.}
In this paper, we aim to analyze the distributed $k$-means clustering problem while we reduce the communication bottleneck between nodes. 
Specifically, we focus on the realistic scenario where nodes communicate with quantized messages and we present a distributed algorithm which operates in an event-triggered fashion and converges in finite time. 
Furthermore, in order to preserve available resources, nodes are able to determine whether the algorithm converged so as to terminate their operation. 
The main contributions of our paper are the following.
\begin{list5}
\item We present a novel distributed algorithm for solving the $k$-means clustering problem. 
The algorithm allows nodes to calculate in a distributed fashion a set of centroids that minimize the sum of squares within every cluster. 
During its operation, nodes exchange quantized messages of finite length with their neighboring nodes. 
We show that our proposed algorithm converges in a deterministic manner after a finite number of time steps, and calculates the \textit{exact} result (represented as the ratio of two quantized values) without introducing any final error (e.g., due to quantization or due to asymptotic convergence). 
Furthermore, the algorithm's operation relies on (i) calculating the average of the observations of every cluster, and (ii) utilizing a distributed stopping strategy to determine whether convergence has been achieved and thus terminate the operation. 
For this reason, we present (i) a novel algorithm for quantized average consensus for the case where each node's initial state is a vector (see Algorithm~\ref{algorithm_1_det_av_multi}), and (ii) a novel distributed stopping mechanism in order to terminate the operation of the proposed algorithm in a finite number of time steps (and hence solve the $k$-means clustering problem in finite time). 
Note that to the authors knowledge, this is the first algorithm for solving the $k$-means clustering problem in a fully distributed manner for the case where the operation of the nodes relies on quantized communication; see Algorithm~\ref{algorithm1}. 
\item We calculate a deterministic upper bound on the required time steps for the convergence of our algorithm. 
Our bound depends on the network structure and the number of centroid calculations; see Theorem~\ref{thm:cluster_convergence}. 
\item We demonstrate the operation of our algorithm via various simulations. 
Furthermore, we compare our algorithm's performance against other $k$-means clustering algorithms; see Section~\ref{sec:results}. 
\end{list5}
The operation of our proposed algorithm relies on consecutive executions of a quantized average consensus algorithm along with a distributed stopping mechanism. 
More specifically, initially each node assigns its observation to the cluster characterized by the nearest centroid. 
It executes a quantized average consensus algorithm to calculate the new centroid values. 
Then, it utilizes a distributed stopping mechanism in order to determine whether the new centroid values have been calculated. 
This allows our algorithm to calculate the optimal centroid values in finite time. 
Furthermore, the state of each node is represented as a fraction of two quantized values. 
This characteristic allows each node to calculate the exact value of each centroid without any error. 
As a result, the proposed algorithm is able to calculate the \textit{exact} local optimal solution of the $k$-means clustering problem. 

The current literature comprises of centralized or distributed algorithms whose operation with real values increases bandwidth and processing requirements and leads to approximate solutions. 
Our paper is a major departure from the current literature since the operation of each node relies on quantized communication. 
Utilization of quantized values allows more efficient usage of network resources, and leads to calculation in finite time of the exact solution without any error. 
Therefore, our proposed distributed algorithm introduces a novel approach for data clustering with efficient (quantized) communication.


%
%
%
%
\section{Mathematical Notation}\label{sec:preliminaries}


\textbf{Graph Theoretic Notions. }
The sets of real, rational, and integer numbers are denoted by $ \mathbb{R}, \mathbb{Q}$, and $\mathbb{Z}$, respectively. 
The symbol $\mathbb{Z}_{\geq 0}$ ($\mathbb{Z}_{>0}$) denotes the set of nonnegative (positive) integer numbers. 
The set $\mathbb{Z}_{\leq 0}$ ($\mathbb{Z}_{<0}$) denotes the set of nonpositive (negative) integer numbers. 
For any real number $a \in \mathbb{R}$, $\lfloor a \rfloor$ denotes the greatest integer less than or equal to $a$ (i.e., the floor of $a$), and $\lceil a \rceil$ denotes the least integer greater than or equal to $a$ (i.e., the ceiling of $a$). 
Vectors are denoted by small letters, and matrices are denoted by capital letters. 
The transpose of matrix $A$ is denoted by $A^T$. 
For matrix $A\in \mathbb{R}^{n\times n}$, $A_{(ij)}$ denotes the entry at row $i$ and column $j$. 
For a vector $a\in \mathbb{R}^{ n}$, $a_{(i)}$ denotes the entry at position $i$. 
The all-ones vector is denoted as $\mathbf{1}$ and the identity matrix is denoted as $I$ (of appropriate dimensions). 

Let us consider a network of $n$ nodes ($n > 2$) where each node can communicate only with its immediate neighbors. 
The communication topology is captured by a directed graph (digraph) defined as $\mathcal{G}_d = (\mathcal{V}, \mathcal{E})$. 
In digraph $\mathcal{G}_d$, $\mathcal{V} =  \{v_1, v_2, \dots, v_n\}$ is the set of nodes with cardinality $n  = | \mathcal{V} | \geq 2 $, and $\mathcal{E} \subseteq \mathcal{V} \times \mathcal{V} - \{ (v_j, v_j) \ | \ v_j \in \mathcal{V} \}$ is the set of edges (self-edges excluded) with cardinality $m = | \mathcal{E} |$. 
A directed edge from node $v_i$ to node $v_j$ is denoted by $m_{ji} \triangleq (v_j, v_i) \in \mathcal{E}$, and captures the fact that node $v_j$ can receive information from node $v_i$ (but not the other way around). 
We assume that the given digraph $\mathcal{G}_d = (\mathcal{V}, \mathcal{E})$ is \textit{strongly connected}. 
This means that for each pair of nodes $v_j, v_i \in \mathcal{V}$, $v_j \neq v_i$, there exists a directed \textit{path}\footnote{A directed \textit{path} from $v_i$ to $v_j$ exists if we can find a sequence of nodes $v_i \equiv v_{l_0},v_{l_1}, \dots, v_{l_t} \equiv v_j$ such that $(v_{l_{\tau+1}},v_{l_{\tau}}) \in \mathcal{E}$ for $ \tau = 0, 1, \dots , t-1$.} from $v_i$ to $v_j$. 
The diameter $D$ of a digraph is the longest shortest path between any two nodes $v_j, v_i \in \mathcal{V}$. 
The subset of nodes that can directly transmit information to node $v_j$ is called the set of in-neighbors of $v_j$ and is represented by $\mathcal{N}_j^- = \{ v_i \in \mathcal{V} \; | \; (v_j,v_i)\in \mathcal{E}\}$. 
The cardinality of $\mathcal{N}_j^-$ is called the \textit{in-degree} of $v_j$ and is denoted by $\mathcal{D}_j^-$. 
The subset of nodes that can directly receive information from node $v_j$ is called the set of out-neighbors of $v_j$ and is represented by $\mathcal{N}_j^+ = \{ v_l \in \mathcal{V} \; | \; (v_l,v_j)\in \mathcal{E}\}$. 
The cardinality of $\mathcal{N}_j^+$ is called the \textit{out-degree} of $v_j$ and is denoted by $\mathcal{D}_j^+$. 


\textbf{Node Operation.}
We assume that each node $v_j$ can directly transmit messages to each out-neighbor; however, it cannot necessarily receive messages (at least not directly) from them. 
In the proposed distributed algorithm, each node $v_j$ assigns a \textit{unique order} in the set $\{0,1,..., \mathcal{D}_j^+ -1\}$ to each of its outgoing edges $m_{lj}$, where $v_l \in \mathcal{N}^+_j$. 
More specifically, the order of link $(v_l,v_j)$ for node $v_j$ is denoted by $P_{lj}$ (such that $\{P_{lj} \; | \; v_l \in \mathcal{N}^+_j\} = \{0,1,..., \mathcal{D}_j^+ -1\}$). 
This unique predetermined order is used during the execution of the proposed algorithm as a way of allowing node $v_j$ to transmit messages to its out-neighbors in a \textit{round-robin}\footnote{When executing the deterministic protocol, each node $v_j$ transmits to its out-neighbors, one at a time, by following the predetermined order. The next time it transmits to an out-neighbor, it continues from the outgoing edge it stopped the previous time and cycles through the edges in a round-robin fashion.} fashion. 


\section{Preliminaries on Distributed Coordination}\label{sec:DistributedCoordination}


The distributed $\max$-consensus algorithm calculates the maximum value of the network in a finite number of time steps \cite{2008:Cortes}. 
The intuition of the algorithm is the following: every node $v_{j}$ in the network, performs the following update rule:
\begin{align}
x_j[\mu+1] = \max_{v_{i}\in \mathcal{N}_j^{-} \cup \{v_{j}\}}\{ x_i[\mu] \}.
\end{align}
The $\max$-consensus algorithm converges to the maximum value among all nodes in a finite number of steps $s$, where $s \leq D$ (see, e.g., \cite[Theorem 5.4]{2013:Giannini}). 
Note here that similar results hold also for the $\min$-consensus algorithm.

\section{Problem Formulation}\label{sec:probForm}

\subsection{$k$-means Clustering}

The problem we present in this paper is borrowed from \cite{2015:Oliva}, but is adjusted in the context of quantized communication over directed networks. 
Specifically, let us consider a set of $n$ observations $x_1, \ldots, x_n$, where $x_i \in \mathbb{R}^d$ for $i \in \{ 1, 2, ..., n \}$.  
Each observation $x_i$ is assigned to each node $v_i$, respectively. 
In the $k$-means clustering problem, we want to partition the $n$ observations into $k$ sets (where $k \leq n$) or {\it clusters} $\mathcal{C} = \{ C_1, \ldots, C_k \}$ so we can minimize the sum of squares within every cluster. 
Specifically, we want to find a set of centroids $\textbf{c}_1, \ldots, \textbf{c}_k$ (where $\textbf{c}_{\gamma} \in \mathbb{R}^d$, for $\gamma \in \{ 1, 2, ..., k \}$), each associated to a cluster, which solve the following optimization problem: 
\begin{align}
\mathcal{D} & = \argmin_{C} \sum_{i=1}^k \sum_{j=1}^n r_{ij} ||x_j-\textbf{c}_i||^2 , \label{original_optptob} \\
& \text{s.t.} \ \ \sum_{i=1}^k r_{ij} = 1 , \ \text{for all} \ j = 1,2,...n \ \text{and} \\
 & r_{ij} \in \{ 0, 1 \} , \qedhere
\qedhere
\end{align}
where $r_{ij} \in \mathbb{Z}_{\geq 0}^{n \times k}$.  
Note here that $r_{ij} = 1$, means that node $v_j$ belongs in cluster $C_i$ ($r_{ij} = 0$ otherwise). 

The problem in \eqref{original_optptob} is hard to solve exactly when $n$ and $k$ are large\footnote{The problem is NP-hard in general Euclidean space $\mathbb{R}^d$, even for $2$ clusters \cite{2009:Popat} and for a general number of clusters $k$, even in the plane \cite{2009:Mahajan}.}, thus calling for approximate solutions.
In particular, the $k$-means algorithm represents a successful strategy to compute a local optimal solution to the above problem.
The intuition of the $k$-means algorithm is that it starts with a random set of $k$ centroids $\textbf{c}_1(1), \ldots, \textbf{c}_k(1)$, and alternates at each step between an {\it assignment} and a {\it refinement} phase.

\textbf{Assignment phase:} Each observation $x_\lambda$ is assigned to the set characterized by the nearest centroid, i.e.:
\begin{equation}\label{Assign}
C_i(T) = \{ x_\lambda : ||x_\lambda-\textbf{c}_i||^2 \leq ||x_\lambda-\textbf{c}_j||^2, \ i,j \in [1, k] \} 
\end{equation}

\textbf{Refinement phase}: Each centroid $\textbf{c}_i(T+1)$ is updated as:
\begin{equation}
\textbf{c}_i(T+1)= \frac{\sum_{v_j \in C_i(T)} x_j}{| C_i(T) |}
\end{equation}
The two steps are iterated until convergence (i.e., if the centroids no longer change) or up to a maximum of $M$ iterations. 

\begin{figure}[t]
\begin{center}
\includegraphics[width=3.5in]{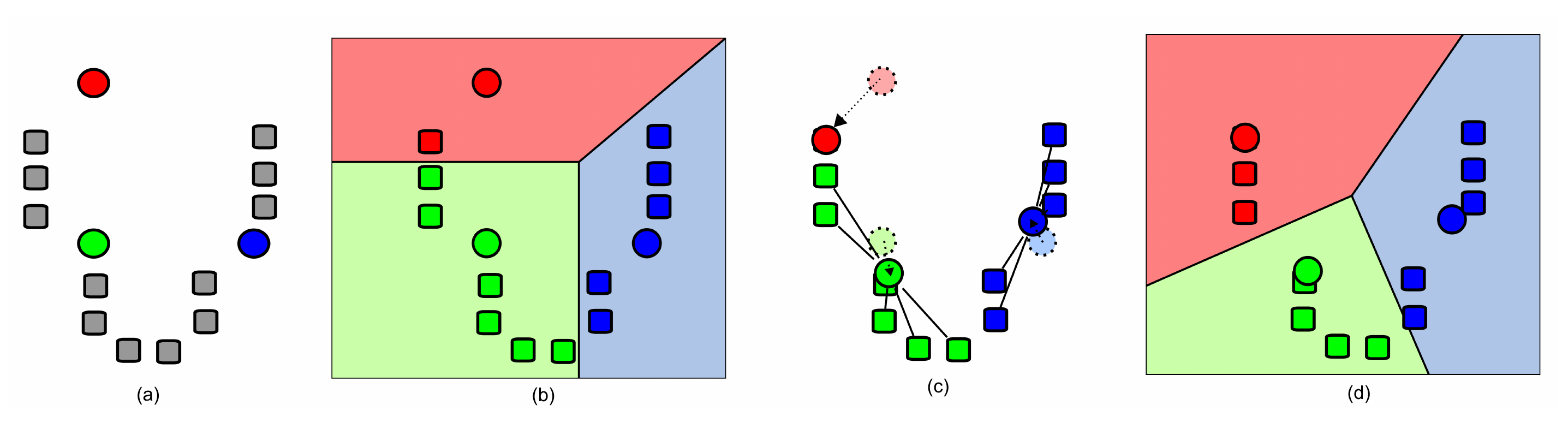}
\caption{Example of execution of $k$-means algorithm (source; Wikipedia Commons available under GNU Free Documentation License v. 1.2).}
 \label{fig:kmeans}
\end{center}
\end{figure}

Fig.~\ref{fig:kmeans} is an example of execution of the algorithm for a set of $n=12$ observations in $\mathbb{R}^2$ and for $k=3$.
Fig.~\ref{fig:kmeans}~(a) shows with circles the initial centroids.  Fig.~\ref{fig:kmeans}~(b) and Fig.~\ref{fig:kmeans}~(c) report the assignment and refinement phases for the first step. 
Fig.~\ref{fig:kmeans}~(d) depicts the assignment phase for the second step. 

The $k$-means algorithm is known to converge to a local optimum value, while there is no guarantee to converge to the global optimum \cite{1967:MacQueen}.
However, given the complexity of the problem at hand, the $k$-means algorithm is de facto the most diffused heuristic algorithm: indeed ``ease of implementation, simplicity, efficiency, and empirical success are the main reasons for its popularity'' \cite{2010:Jain}. 
Furthermore, note that the convergence of the algorithm strongly depends on the initial choice of the centroids. 
Therefore, a common practice is to execute the algorithm several times -- each time with different initial conditions -- and select the best solution. 

\subsection{Modification of $k$-means Clustering Problem: Finite Time $k$-means Clustering with Quantized Communication}\label{modified_kmeans}
In this paper, we develop a distributed algorithm that allows nodes to find a locally optimal solution to the problem \textbf{P1} presented below, while transmitting quantized information via available communication links. 

\textbf{P1.} Consider a static strongly connected digraph $\mathcal{G}_d$, where each node $v_j$ is endowed with a quantized state $x_j \in \mathbb{R}^d$. 
We aim at developing a distributed algorithm which calculates in a distributed fashion a set of centroids $\textbf{c}_1, \ldots, \textbf{c}_k$ (where $\textbf{c}_{\gamma} \in \mathbb{R}^d$, for $\gamma \in \{ 1, 2, ..., k \}$) and association variables $r_{ij}$, which represent a locally optimal solution to the following optimization problem: 
    \begin{align} 
    \mathcal{D} & = \argmin_{C} \sum_{i=1}^k \sum_{j=1}^n r_{ij} ||x_j-\textbf{c}_i||^2 , \label{optptob} \\
    & \text{s.t.} \ \ \sum_{i=1}^k r_{ij} = 1 , \ \text{for all} \ j = 1,2,...n \ \label{assign_cond} \\
    & r_{ij} \in \{ 0, 1 \} \label{assign_one}, \qedhere
    \qedhere
    \end{align}
    where $r_{ij} \in \mathbb{Z}_{\geq 0}^{n \times k}$.  
During the proposed algorithm each node transmits quantized information. 
The proposed algorithm converges in a finite number of time steps, upper bounded by a polynomial function which depends on the communication network. 
Each node ceases transmissions once convergence has been achieved. 

\section{Finite Time $k$-means Clustering with Quantized Communication}
\label{sec:distr_algo}

In this section we propose a distributed algorithm which solves problem \textbf{P1} in Section~\ref{modified_kmeans}. 
We first present an extended version of the algorithm in \cite{2020:Rikos_TAC_Mass_Acc} which is important for our subsequent development.

\subsection{Multidimensional Deterministic Exact Quantized Average Consensus}\label{Multidim_Prel_Aver}

In this section, we present an extended version of the deterministic algorithm in \cite{2020:Rikos_TAC_Mass_Acc}. 
In this version, each node is able to calculate the exact average of the initial states in a deterministic fashion after a finite number of time steps for the case where the state of each node is an integer vector (i.e., $y_j[\mu] \in \mathbb{Z}^{d}$, where $\mu, d \in \mathbb{Z}_{> 0}$). 
The proposed algorithm is detailed as Algorithm~\ref{algorithm_1_det_av_multi} below. 

\noindent
\begin{varalgorithm}{1}
\caption{Multidimensional Deterministic Exact Quantized Average Consensus Algorithm}
\noindent \textbf{Input:} A strongly connected digraph $\mathcal{G}_d = (\mathcal{V}, \mathcal{E})$ with $n=|\mathcal{V}|$ nodes and $m=|\mathcal{E}|$ edges. 
Each node $v_j \in \mathcal{V}$ has an initial quantized state $y_j[1] \in \mathbb{Z}^{d}$. \\ 
\textbf{Initialization:} Every node $v_j \in \mathcal{V}$ does the following: 
\begin{list4}
\item assigns to each of its outgoing edges $v_l \in \mathcal{N}^+_j$ a \textit{unique order} $P_{lj}$ in the set $\{0,1,..., \mathcal{D}_j^+ -1\}$; 
\item sets $tr^{(j)} = 0$ and $e = tr^{(j)}$; 
\item sets $z_j[1] = 1$, $z^s_j[1] = z_j[1]$ and $y^s_j[1] = y_j[1]$ (which means that $q^s_j[1] = y^s_j[1] / z^s_j[1]$); 
\item chooses out-neighbor $v_l \in \mathcal{N}_j^+$ according to the predetermined order $P_{lj}$ (initially, it chooses $v_l \in \mathcal{N}_j^+$ such that $P_{lj}=0$) and transmits $z_j[1]$ and $y_j[1]$ to this out-neighbor. Then, it sets $y_j[1] = 0$, $z_j[1] = 0$, $\text{pass}_j = 0$; 
\item sets $tr^{(j)} = tr^{(j)} + 1$ and $e = tr^{(j)} \mod \mathcal{D}^+_j$; 
\end{list4} 
\textbf{Iteration:} For $\mu=1,2,\dots$, each node $v_j \in \mathcal{V}$, does the following: 
\begin{list4} 
\item receives $y_i[\mu]$ and $z_i[\mu]$ from its in-neighbors $v_i \in \mathcal{N}_j^-$ and sets 
$$
y_j[\mu+1] = y_j[\mu] + \sum_{v_i \in \mathcal{N}_j^-} w_{ji}[\mu]y_i[\mu] ,
$$
and 
$$
z_j[\mu+1] = z_j[\mu] + \sum_{v_i \in \mathcal{N}_j^-} w_{ji}[\mu]z_i[\mu] ,
$$
where  $w_{ji}[\mu] = 0$ if no message is received (otherwise $w_{ji}[\mu] = 1$); 
\item sets $\text{pass}_j = 1$, sets $\text{dim}_j = 1$; 
\item \underline{Event Trigger Conditions:} \textbf{while} $\text{dim}_j \leq d$ \textbf{then}
\\ C$1$: \textbf{if} $z_j[\mu+1] > z^s_j[\mu]$ \textbf{break}; 
\\ C$2$: \textbf{if} $z_j[\mu+1] < z^s_j[\mu]$ \textbf{sets} $\text{dim}_j = 0$, \textbf{break}; 
\\ C$3$: \textbf{if} $z_j[\mu+1] = z^s_j[\mu]$ and $y_{j(\text{dim}_j)}[\mu+1] > y^s_{j(\text{dim}_j)}[\mu]$ \textbf{break}; 
\\ C$4$: \textbf{if} $z_j[\mu+1] = z^s_j[\mu]$ and $y_{j(\text{dim}_j)}[\mu+1] = y^s_{j(\text{dim}_j)}[\mu]$ \textbf{sets} $\text{dim}_j = \text{dim}_j + 1$; 
\\ C$5$: \textbf{if} $z_j[\mu+1] = z^s_j[\mu]$ and $y_{j(\text{dim}_j)}[\mu+1] < y^s_{j(\text{dim}_j)}[\mu]$ \textbf{sets} $\text{dim}_j = 0$, \textbf{break}; 
\item \textbf{if} $\text{pass}_j = 1$: 
\begin{list4a}
\item sets $z^s_j[\mu+1] = z_j[\mu+1]$, $y^s_j[\mu+1] = y_j[\mu+1]$, 
$$
q^s_j[\mu+1] = \frac{y^s_j[\mu+1]}{z^s_j[\mu+1]} .
$$
\item transmits $z_j[\mu+1]$ and $y_j[\mu+1]$ towards out-neighbor $v_{\lambda} \in \mathcal{N}_j^+$ for which $P_{\lambda j} = e$ and it sets $y_j[\mu+1] = 0$ and $z_j[\mu+1] = 0$. 
Then it sets $tr^{(j)} = tr^{(j)} + 1$ and $e = tr^{(j)} \mod \mathcal{D}^+_j$. 
\end{list4a}
\end{list4}
\textbf{Output:} \eqref{alpha_q_no_oscill}  holds for every $v_j \in \mathcal{V}$. 
\label{algorithm_1_det_av_multi}
\end{varalgorithm}

The intuition of Algorithm~\ref{algorithm_1_det_av_multi} is the following. 
Each node $v_j$ receives the mass variables $y_i[k]$ and $z_i[k]$ from its in-neighbors $v_i \in \mathcal{N}_j^-$ and sums them along with its stored mass variables ($y_j[k]$ and $z_j[k]$). 
During the event-triggered conditions C$1$ - C$5$, node $v_j$ compares each element of the received vectors against its stored vectors. 
According to the event-triggered conditions, it decides whether it will update its state variables and will perform transmission towards one of its out-neighbors. 
If it performs a transmission, it sets its stored mass variables equal to zero and repeats the procedure. 

\begin{defn}\label{Definition_Quant_Av}
The system is able to achieve \textit{exact} quantized average consensus in the form of a quantized fraction if, for every $v_j \in \mathcal{V}$, there exists $\mu_0 \in \mathbb{Z}_+$ so that for every $v_j \in \mathcal{V}$ we have
\begin{equation}\label{alpha_q_no_oscill}
q^s_j[\mu] = \frac{\sum_{l=1}^{n}{y_l[0]}}{n} , 
\end{equation}
for $\mu \geq \mu_0$, where $q$ is the real average of the initial states defined as 
\begin{equation}\label{real_av}
q = \frac{\sum_{l=1}^{n}{y_l[0]}}{n} .
\end{equation}
\end{defn}

Let us now consider the following setup. 

{\it Setup~$1$:} Consider a strongly connected digraph $\mathcal{G}_d = (\mathcal{V}, \mathcal{E})$ with $n=|\mathcal{V}|$ nodes and $m=|\mathcal{E}|$ edges. 
Each node $v_j \in \mathcal{V}$ has an initial quantized state $y_j[0] \in \mathbb{Z}^{d}$. 
During the execution of the Algorithm~\ref{algorithm_1_det_av_multi}, at time step $\mu_1$, there is at least one node $v_{j'} \in \mathcal{V}$, for which 
\begin{equation}\label{great_z_prop1_det_multi}
z_{j'}[\mu_1] \geq z_i[\mu_1], \ \forall v_i \in \mathcal{V}.
\end{equation}
Then, among the nodes $v_{j'}$ for which (\ref{great_z_prop1_det_multi}) holds, there is at least one node $v_j$ for which 
\begin{equation}\label{great_z_prop2_det_multi}
y_{j(\text{dim}_j)}[\mu_1] \geq y_{j'(\text{dim}_j)}[\mu_1] , \ v_j, v_{j'} \in \{ v_i \in \mathcal{V} \ | \ (\ref{great_z_prop1_det_multi}) \ \text{holds} \} , 
\end{equation}
for every $\text{dim}_j \in \{1, 2, ..., d\}$.  
For notational convenience we will call the mass variables of node $v_j$ for which (\ref{great_z_prop1_det_multi}) and (\ref{great_z_prop2_det_multi}) hold as the ``leading mass'' (or ``leading masses''). 

In the following theorem we present the deterministic convergence of Algorithm~\ref{algorithm_1_det_av_multi}. 
The proof of the theorem is similar to Proposition~$1$ in \cite{2020:Rikos_TAC_Mass_Acc} and is omitted. 

\begin{theorem}[]
\label{Conver_Quant_Av_multi}
Under {\it Setup~$1$} we have that the execution of Algorithm~\ref{algorithm_1_det_av_multi} allows each node $v_j \in \mathcal{V}$ to reach quantized average consensus after a finite number of steps $\mathcal{S}_t$, bounded by $\mathcal{S}_t \leq nm^2$. 
\end{theorem}

\begin{remark}
For developing our results in this paper, we rely on the operation of Algorithm~\ref{algorithm_1_det_av_multi}. 
As mentioned previously, this algorithm allows deterministic convergence after a finite number of time steps as shown in Theorem~\ref{Conver_Quant_Av_multi}. 
Please note that we can also rely on the operation of the algorithm in \cite{2018:RikosHadj} for developing our results. 
The operation of \cite{2018:RikosHadj} is simpler compared to Algorithm~\ref{algorithm_1_det_av_multi}. 
However, \cite{2018:RikosHadj} does not exhibit deterministic convergence but rather converges with high probability to the exact real average $q$ in \eqref{real_av} after a finite number of time steps. 
\end{remark}

\subsection{Finite Time $k$-means Clustering Algorithm with Quantized Communication}\label{distr_algo_kmeans_quant}

We now present a distributed algorithm which solves Problem \textbf{P1} presented in Section~\ref{modified_kmeans}. 
The proposed algorithm is detailed as Algorithm~\ref{algorithm1} below and allows each node in the network to calculate, while processing and transmitting quantized messages, in a finite number of time steps, a set of centroids $\textbf{c}_1, \ldots, \textbf{c}_k$, each associated to a cluster, which fulfill \eqref{optptob}. 
Furthermore, each node is able to determine whether convergence has been achieved and proceed to cease transmissions. 
To solve the $k$-means clustering problem, we make the following two assumptions which are necessary for the operation of Algorithm~\ref{algorithm1}. 

\begin{assum}\label{Diam_known}
Every node $v_j \in \mathcal{V}$ knows the diameter of the network $D$ (or an upper bound $D'$). 
\end{assum}
\begin{assum}\label{leader_choosing}
Each node $v_j$ knows the initial set of centroids $C[0] = [ \textbf{c}_{1}[0], \textbf{c}_{2}[0], ..., \textbf{c}_{k}[0] ] \in \mathbb{R}^{d \times k} $ ($k < n$). 
\end{assum}

Assumption~\ref{Diam_known} is a necessary condition for the $\min$- and $\max$-consensus algorithm, so that each node $v_j$ is able to determine whether convergence has been achieved and thus our proposed algorithm can terminate. 
Note, however, that this assumption can be relaxed if we utilize the distributed algorithm in \cite{2021:Rikos_Hadj_Johan_2} instead of Algorithm~\ref{algorithm_1_det_av_multi}. 
The algorithm in \cite{2021:Rikos_Hadj_Johan_2} converges to the exact real average in finite time without requiring knowledge of the network diameter $D$. 
Assumption~\ref{leader_choosing} is a necessary condition so that each node can calculate the updated value of the centroids without having to communicate their real values to other nodes (because communication is restricted to quantized values). 

\begin{remark}
Regarding Assumption~\ref{leader_choosing}, previous work in \cite{2015:Oliva} ensures that one node is elected as the leader node and propagates the real values of the centroids to every node. 
Note that in our case, nodes communicate by exchanging quantized values. 
Therefore, each node needs to know the initial set of centroids in order to calculate their updated values without the need of a leader node (which transmits the updated value of the centroids to every node in the network). 
When the initial set of centroids is known only to a certain leader node, then the leader node can propagate the set of centroids to every node if the initial set of centroids are quantized values.
Thus, after $D$ time steps, every node in the network will know the initial centroids and Assumption~\ref{leader_choosing} will be fulfilled. 
\end{remark}

We now describe the main operations of Algorithm~\ref{algorithm1}. 
The initialization involves the following steps: 

\textbf{Initialization. Centroid Selection and Unique Order:} 
Each node $v_j \in \mathcal{V}$ has a quantized state $x_j \in \mathbb{Z}^{d}$. 
Then, it assigns to each of its outgoing edges a \textit{unique order}. 

The iteration involves the following steps: 

\textbf{Iteration - Step~$1$. Cluster Assignment and Labeled Multidimensional Deterministic Exact Quantized Average Consensus:} 
At each step $\mu$, each node $v_j$ assigns its value $x_j$ to the nearest centroid according to \eqref{Assign} (since each node knows the set of initial centroids $C(0) = [ \textbf{c}_{1}(0), \textbf{c}_{2}(0), ..., \textbf{c}_{k}(0) ] \in \mathbb{R}^{d \times k} $ ($k < n$)). 
This means that it sets $r_{\lambda j} = 1$ (where $\textbf{c}_{\lambda}(0)$ is the nearest centroid) with respect to \eqref{assign_cond}. 
Then, it performs $k$ quantized average consensus operations -- each operation is done with the states of the nodes that belong in the same cluster. 
Specifically, each node $v_j$ executes $k$ times in parallel Algorithm~\ref{algorithm_1_det_av_multi} in Section~\ref{Multidim_Prel_Aver}. 
Each execution is done with the nodes $\{ v_i \in \mathcal{V} \ | \ r_{\lambda i} = r_{\lambda j} = 1 \}$ (i.e., that belong in the same cluster). 
In this way, each node calculates the \textit{exact} updated value of every centroid in finite time. 

\textbf{Iteration - Step~$2$. Labeled Distributed Stopping:} 
Every node $v_j$ performs $k$ parallel $\min-$ and $\max-$consensus operations every $D$ time steps as described in Section~\ref{sec:DistributedCoordination}. 
Each operation is done with the states of the nodes that belong in the same cluster (i.e., with nodes $v_i, v_j$ for which $r_{\lambda j} = r_{\lambda i} = 1$). 
In this way, each node is able to determine whether convergence has been achieved and the updated set of centroids $\textbf{c}_{\gamma}[T + 1]$, for every $\gamma \in \{1, 2, ..., k \}$, has been calculated. 

\textbf{Iteration - Step~$3$. Centroid Update, Cluster Assignment and Algorithm Termination:} 
Once all $k$ executions of Algorithm~\ref{algorithm_1_det_av_multi} have converged, each node $v_j$ updates the stored set of centroids. 
Then, it assigns its value $x_j$ to the nearest updated centroid according to \eqref{Assign}. 
It checks if the previous centroid values are equal to the new centroid values. 
If this condition holds for each node $v_j$, the operation of the algorithm is terminated. 
Otherwise, the iteration is repeated. 

The flowchart for the operation of each node during Algorithm~\ref{algorithm1} is shown in Fig.~\ref{operation_flowchart}. 
In the flowchart we can see that initially each node assigns a unique order to its outgoing edges and it also assigns its value to the cluster characterized by the nearest centroid. 
Then, for each of the $k$ clusters it executes (i) ``Labeled Multidimensional Deterministic Exact Quantized Average Consensus'' (shown in Algorithm~\ref{algorithm_1_det_av_multi} for the case where we have one cluster $k=1$), and (ii) ``Labeled Distributed Stopping''. 
This allows the calculation of the new centroid values. 
Then, it checks if the previous centroid values are equal to the new centroid values. 
If this condition holds, the algorithm has converged and every node terminates its operation. Otherwise, the process is repeated. 

\begin{figure}[t]
    \centering
    \includegraphics[width=5cm]{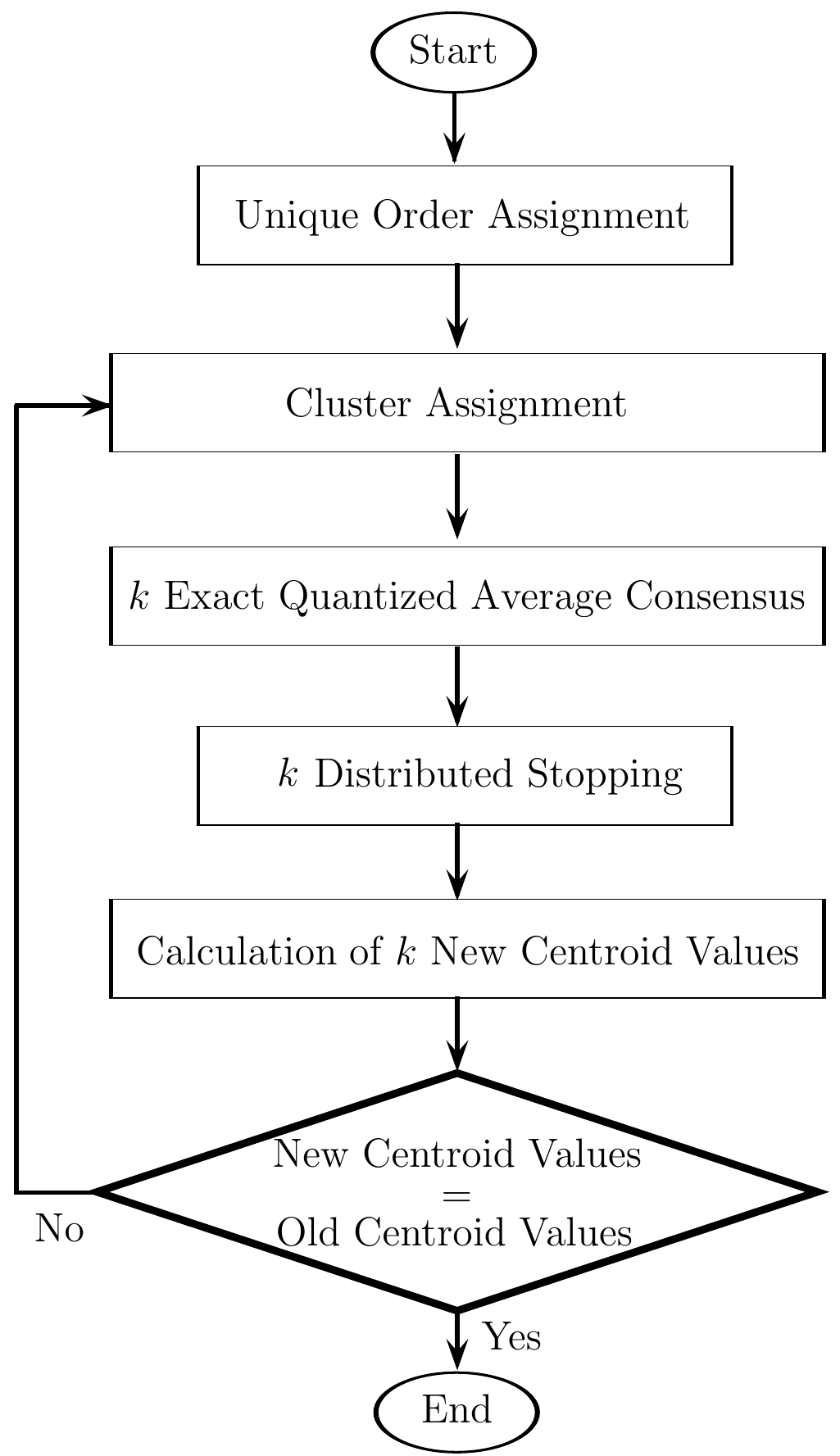}
    \caption{Flowchart for operation of each node during Algorithm~\ref{algorithm1}.}
    \label{operation_flowchart}
\end{figure}

\begin{varalgorithm}{2}
\caption{Finite Time Quantized $k$-means Algorithm}
\noindent \textbf{Input:} A strongly connected digraph $\mathcal{G}_d = (\mathcal{V}, \mathcal{E})$ with $n=|\mathcal{V}|$ nodes and $m=|\mathcal{E}|$ edges. 
Each node $v_j \in \mathcal{V}$ has an initial quantized state $x_j \in \mathbb{Z}^{d}$, and has knowledge of $D$. Each node $v_j$ knows the number of clusters $k < n$ and the initial centroids $C[0] = [ \textbf{c}_{1}[0], \textbf{c}_{2}[0], ..., \textbf{c}_{k}[0] ] \in \mathbb{R}^{d \times k} $. \\
\textbf{Initialization:} Each node $v_j$ sets $\text{flag}_j = 0$ and assigns to each of its outgoing edges $v_l \in \mathcal{N}^+_j$ a \textit{unique order} $P_{lj}$ in the set $\{0,1,..., \mathcal{D}_j^+ -1\}$. \\
\textbf{Iteration:} For $\mu = 1, 2, \dots$, each node $v_j \in \mathcal{V}$, does the following: 
\begin{list4} 
\item \textbf{while} $\text{flag}_j = 0$ \textbf{then} 
\begin{list4a} 
\item sets $r_{j \lambda} = 1$, where $|| x_j - c_{\lambda}[\mu] || \leq || x_j - c_{\gamma}[\mu] || $, where $\gamma \in \{1, 2, ..., k \} \setminus \{ \lambda \}$; 
\item sets $x^{cl}_j[\mu] = x_j$, for $cl = \lambda$, where $r_{j \lambda} = 1$, and $x^{cl}_j[\mu] = 0$, for $cl = \{ 1, 2, ..., k \} \setminus \{ \lambda \}$; 
\item calls Algorithm~\ref{algorithm_max_1a};
\end{list4a} 
\item \textbf{if} 
\begin{equation}\label{stop_cond}
\textbf{c}_{\gamma}(T+1) = \textbf{c}_{\gamma}(T), \ \text{for every} \ \gamma \in \{ 1, 2, ..., k \} ,
\end{equation}
\textbf{then} $\text{flag}_j = 1$; 
\end{list4}
\textbf{Output:} \eqref{optptob}, \eqref{assign_cond}, \eqref{assign_one} hold for every $v_j \in \mathcal{V}$. 
\label{algorithm1}
\end{varalgorithm}

\begin{varalgorithm}{2.A}
\caption{Extended Labeled Multidimensional Deterministic Quantized Average Consensus with Labeled Distributed Stopping}
\noindent \textbf{Input:} $D$, $T$, $\mu$, $x^{cl}_j[\mu]$ for $cl = \{ 1, 2, ..., k \}$, $P_{lj}$ for every $v_l \in \mathcal{N}_j^-$; \\ 
\textbf{Initialization:} $y^{cl}_j[1] = x^{cl}_j[\mu]$ for $cl = \{ 1, 2, ..., k \}$; \\ 
\textbf{Iteration:} For $\mu'=1,2,\dots$, each node $v_j \in \mathcal{V}$, does the following: 
\begin{list4} 
\item \textbf{while} $\text{flag}_j = 0$ \textbf{then} 
\begin{list4a}
\item \textbf{if} $\mu' \mod D = 1$ \textbf{then} sets $M^{cl}_j = x^{cl}_j[\mu'] / z^{cl}_j[\mu']$, $m^{cl}_j = x^{cl}_j[\mu'] / z^{cl}_j[\mu'] $, where $z_j[\mu']=1$, $cl = \{ 1, 2, ..., k \}$; 
\item broadcasts $M^{cl}_j$, $m^{cl}_j$, $cl = \{ 1, 2, ..., k \}$, to every $v_{j} \in \mathcal{N}_j^+$; 
\item receives $M^{cl}_i$, $m^{cl}_i$ from every $v_{i} \in \mathcal{N}_j^-$, $cl = \{ 1, 2, ..., k \}$; 
\item sets $M^{cl}_j = \max_{v_{i} \in \mathcal{N}_j^-\cup \{ v_j \}} \ M^{cl}_i$, $m^{cl}_j = \min_{v_{i} \in \mathcal{N}_j^-\cup \{ v_j \}} \ m^{cl}_j$; 
\item executes Iteration Steps of Algorithm~\ref{algorithm_1_det_av_multi} for each $cl = \{ 1, 2, ..., k \}$ with initial state $y^{cl}_j[\mu]$; 
\item receives $z^{cl}_i[\mu']$, $y^{cl}_i[\mu']$ from $v_i \in \mathcal{N}_j^-$ and sets 
\begin{equation}\label{no_del_eq_y_no_oscil}
y^{cl}_j[\mu'+1] = y^{cl}_j[\mu'] + \sum_{v_i \in \mathcal{N}_j^-}  w_{ji}[\mu'] \ y^{cl}_i[\mu'] , 
\end{equation} 
\begin{equation}\label{no_del_eq_z_no_oscil}
z^{cl}_j[\mu'+1] = z^{cl}_j[\mu'] + \sum_{v_i \in \mathcal{N}_j^-} w_{ji}[\mu'] \ z^{cl}_i[\mu'] ,
\end{equation}
where $w_{ji}[\mu'] = 1$ if node $v_j$ receives a message from $v_i \in \mathcal{N}_j^-$ at iteration $\mu'$ (otherwise $w_{ji}[\mu'] = 0$); 
\item \textbf{if} $\mu' \mod D = 0$ \textbf{then}, \textbf{if} $M^{cl}_j = m^{cl}_j$, for every $cl = \{ 1, 2, ..., k \}$ \textbf{then} sets $\textbf{c}_{cl}[T + 1] = q^{s, cl}_j[\mu']$ for every $cl = \{ 1, 2, ..., k \}$ and sets $\text{flag}_j = 1$. 
\end{list4a}
\end{list4}
\textbf{Output:} $\textbf{c}_{cl}[T + 1]$ for every $cl = \{ 1, 2, ..., k \}$. 
\label{algorithm_max_1a}
\end{varalgorithm}

Next, we show that, during the operation of Algorithm~\ref{algorithm1}, each node $v_j$ is able to (i) calculate a set of centroids $\textbf{c}_1, \ldots, \textbf{c}_k$ that fulfill \eqref{optptob} after a finite number of time steps, and (ii) terminate its operation once convergence has been achieved. 

\begin{theorem}
\label{thm:cluster_convergence}
Consider a strongly connected digraph $\mathcal{G}_d = (\mathcal{V}, \mathcal{E})$ with $n=|\mathcal{V}|$ nodes and $m=|\mathcal{E}|$ edges. 
Each node $v_j \in \mathcal{V}$ has an initial quantized state $x_j \in \mathbb{Z}^{d}$. 
Each node $v_j$ knows the $k$ initial clusters (where $k < n$) and their initial centroids $C[0] = [ \textbf{c}_{1}[0], \textbf{c}_{2}[0], ..., \textbf{c}_{k}[0] ] \in \mathbb{R}^{d \times k} $. 
During the operation of Algorithm~\ref{algorithm1}, each node $v_j$ is able to address problem \textbf{P1} in Section~\ref{modified_kmeans} after a finite number of time steps $C_t$ bounded by 
\begin{equation}
    C_t \leq T (D + nm^2) ,
\end{equation}
where $T$ is the number of new centroid calculations until \eqref{stop_cond} holds, $D$ is the diameter of network $\mathcal{G}_d$. 
\end{theorem}

\begin{proof}
Similarly to the work in \cite{2015:Oliva}, the operation of Algorithm~\ref{algorithm1} follows the same step as that for the centralized $k$-means algorithm except that every step is performed in a decentralized manner. 
Therefore, Algorithm~\ref{algorithm1} is able to calculate, in a distributed fashion, a set of centroids $\textbf{c}_1, \ldots, \textbf{c}_k$ (where $\textbf{c}_{\gamma} \in \mathbb{R}^d$, for $\gamma \in \{ 1, 2, ..., k \}$), each associated to a cluster, which fulfill \eqref{optptob}, \eqref{assign_cond} and \eqref{assign_one} in finite time. 

During the operation of Algorithm~\ref{algorithm1}, we have that each node in the network executes Algorithms~\ref{algorithm_1_det_av_multi}, \ref{algorithm_max_1a}, until \eqref{stop_cond} holds 
(i.e., the new centroid values are equal to the old centroid values). 
The required number of time steps for convergence of Algorithm~\ref{algorithm_1_det_av_multi} is equal to $n m^2$ (see Theorem~\ref{Conver_Quant_Av_multi}). 
The required number of time steps of the distributed stopping protocol is $D$ since its operation relies on max- and min-consensus (see Section~\ref{sec:DistributedCoordination}). 
Furthermore, we use $T$ to denote the number of new centroid calculations until \eqref{stop_cond} holds. 
As a result, we have that during the operation of Algorithm~\ref{algorithm1} after a finite number of time steps $C_t$ bounded by 
$
    C_t \leq T (D + nm^2), 
$
each node $v_j$ is able to address problem \textbf{P1} in Section~\ref{modified_kmeans}. 
\end{proof}

\subsection{Advantages of Finite Time $k$-means Clustering with Quantized Communication}\label{compar_prev_work}

Compared to \cite{2015:Oliva} and \cite{2017:Zheng_Qin}, the main advantage of Algorithm~\ref{algorithm1} is in the (i) network requirements, (ii) centroid calculation step and, (iii) distributed stopping step. 

In \cite{2015:Oliva} the new centroid values are calculated via a finite time average consensus algorithm which operates over the corresponding cluster. 
Furthermore, \cite{2015:Oliva} assumes that the graph underlying the agents' interaction is undirected (in particular, the finite-time average consensus algorithms adopted are not suitable for directed graphs).
Moreover, in \cite{2015:Oliva} agents need to know on upper bound for the number of nodes in the network.
Finally, in \cite{2015:Oliva} each node processes and transmits real valued messages. Therefore, in case communication is quantized, there are no convergence guarantees; moreover, the messages require more bandwidth than the proposed approach. 

In \cite{2017:Zheng_Qin} the underlying network is modelled as a strongly connected digraph. 
However, in order to calculate the average of the node's values the digraph needs to be weight balanced which is a strong assumption (e.g., see \cite{2014:TCNS_June} and references therein). 
Furthermore, the algorithm requires a normalization step (used to obtain the maximum and minimum value of every component of every observation) and execution of $k$-means++ algorithm to produce the initial centroids. 
Additionally, each node knows an upper bound for the number of nodes in the network. 
Finally, as in \cite{2015:Oliva}, this algorithm does not provide convergence guarantees when communication is quantized and messages require more bandwidth than our proposed approach.

In Algorithm~\ref{algorithm1} the underlying network is modelled as a strongly connected digraph. 
However, in order to calculate the average of the nodes' values, the digraph does not need to be weight balanced. 
Furthermore, in order to terminate its operation, each node knows an upper bound over the network diameter and not on the number of nodes. 
Additionally, each node transmits quantized messages, which do not require a large amount of bandwidth for communication and are more suitable for realistic applications (i.e., nodes need to transmit small messages of finite length). 
Finally, by utilizing Algorithm~\ref{algorithm_1_det_av_multi} (which is an extension of \cite{2020:Rikos_TAC_Mass_Acc}), nodes are able to calculate the \textit{exact} solution in finite time without a final error. 
This means that the set of centroids that minimize \eqref{original_optptob} are calculated exactly. 
As a result, Algorithm~\ref{algorithm1} computes the exact minimum of the sum of squares within every cluster (see Section~\ref{sec:probForm}) (and not an approximation of the minimum due to a final error). 











%
%
%
%

\section{Simulation Results} \label{sec:results}

We now present simulation results to illustrate the behavior of Algorithm~\ref{algorithm1} over several examples. 


\textbf{Evaluation over a Random Network of $100$ Nodes.}
We execute Algorithm~\ref{algorithm1} over a random digraph of $100$ nodes with diameter $D = 4$. 
The $100$ nodes are randomly selected in the region $[50, 50] \times [50, 50]$ with uniform probability.
During the operation of Algorithm~\ref{algorithm1} we partition nodes into $k=3$ clusters and calculate the centroid values that fulfill \eqref{optptob} in finite time. 
In this case, the observations coincide with the agent’s positions (i.e., the agents are clustered according to their position). 
In Fig.~\ref{100nodes_50times50}~(A) the initial positions of the $3$ centroids are marked by red, blue and green crosses, the nodes are marked with circles, and each circle color is the color of the nearest initial centroid (i.e., nodes are marked red, blue, or green color). 
In Fig.~\ref{100nodes_50times50}~(B) we show the trajectories of the centroids which are marked with red, blue, or green lines according to the color of the centroid. 
We can see that after $T = 10$ the centroid values fulfill \eqref{stop_cond}. 
This means that the nodes are able to determine that convergence has been achieved and thus proceed to terminate their operation. 
In Fig.~\ref{obj100nodes_50times50obj} we show the evolution of the Distance Objective Function $F(T)$ defined as
\begin{equation}\label{DistanceObjectiveFunction}
    F(T) = \sum_{i=1}^k \sum_{v_j \in C_i(T)} ||x_j-\textbf{c}_i(T)||^2 , 
\end{equation}
for the network of Fig.~\ref{100nodes_50times50}. 
We can see that $F[T]$ is non-increasing over time. 
Furthermore, we can see that $F[9] = F[10]$. 
This means that for $T = 10$ the centroid values fulfill \eqref{stop_cond} and nodes terminate their operation (also shown in Fig.~\ref{100nodes_50times50}~(B)). 

\begin{figure}[t]
    \centering
    \includegraphics[width=7.0cm]{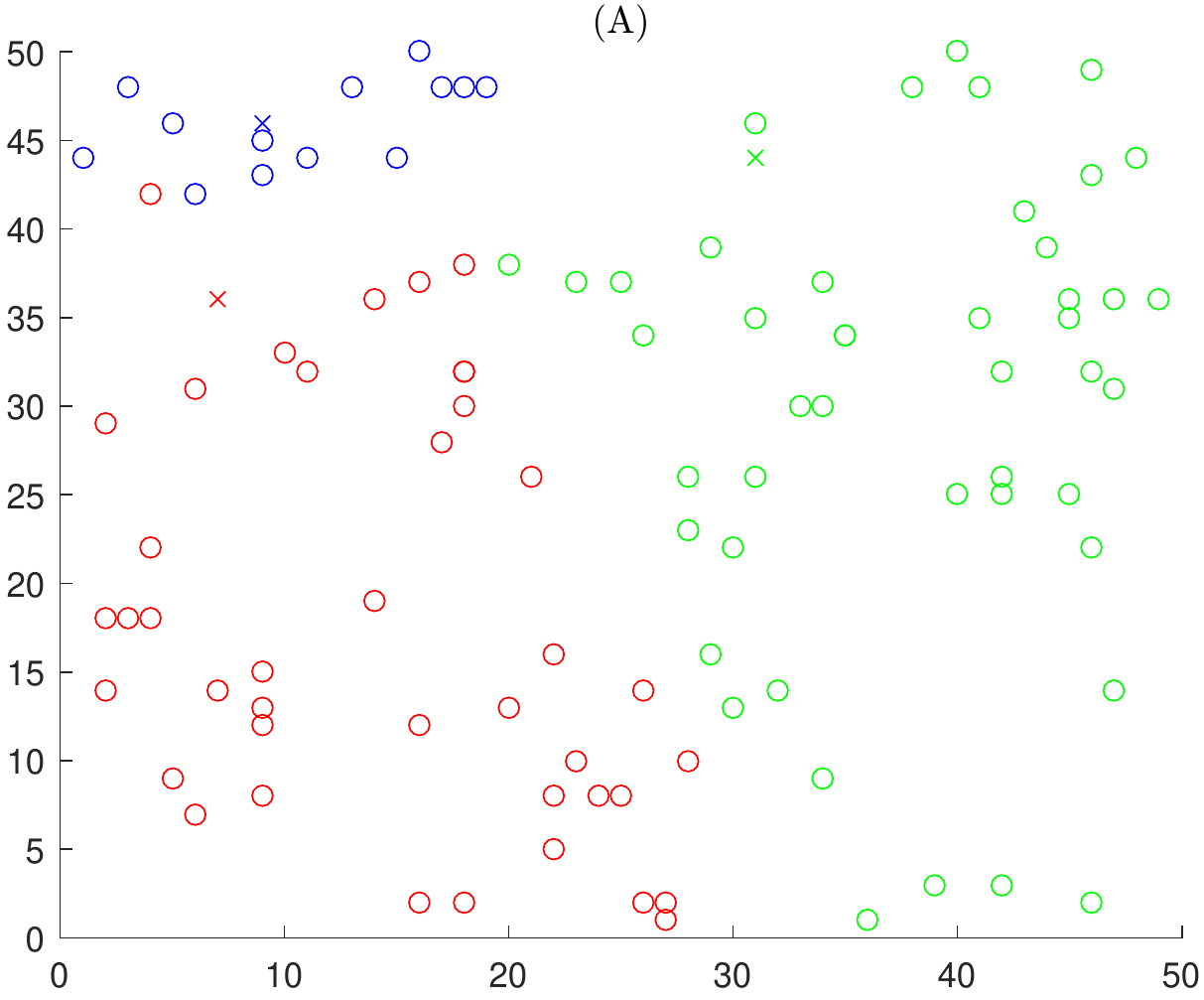}\\ \vspace{.3cm}
    \includegraphics[width=7.0cm]{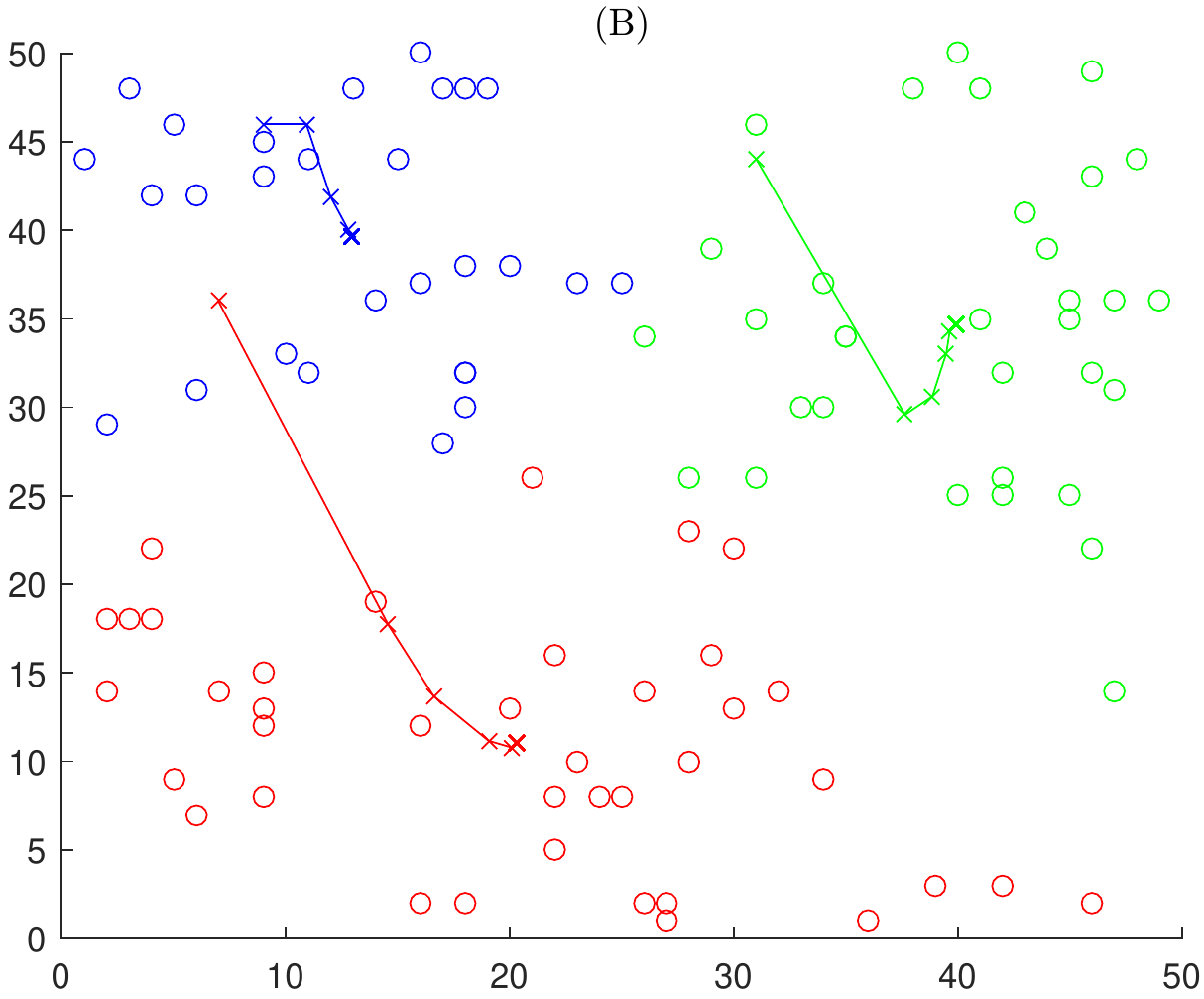}
    \caption{Execution of Algorithm~\ref{algorithm1} over a random directed network with $100$ nodes and $3$ clusters. Positions of nodes are marked with circles and centroid positions are marked with red, blue and green crosses. (A) Initial centroid positions, and initial position and cluster assignment for every node. (B) Centroid trajectories, final centroid positions, and final cluster assignment for every node.} 
    \label{100nodes_50times50}
\end{figure}

\begin{figure}[t]
    \centering
    \includegraphics[width=7.5cm]{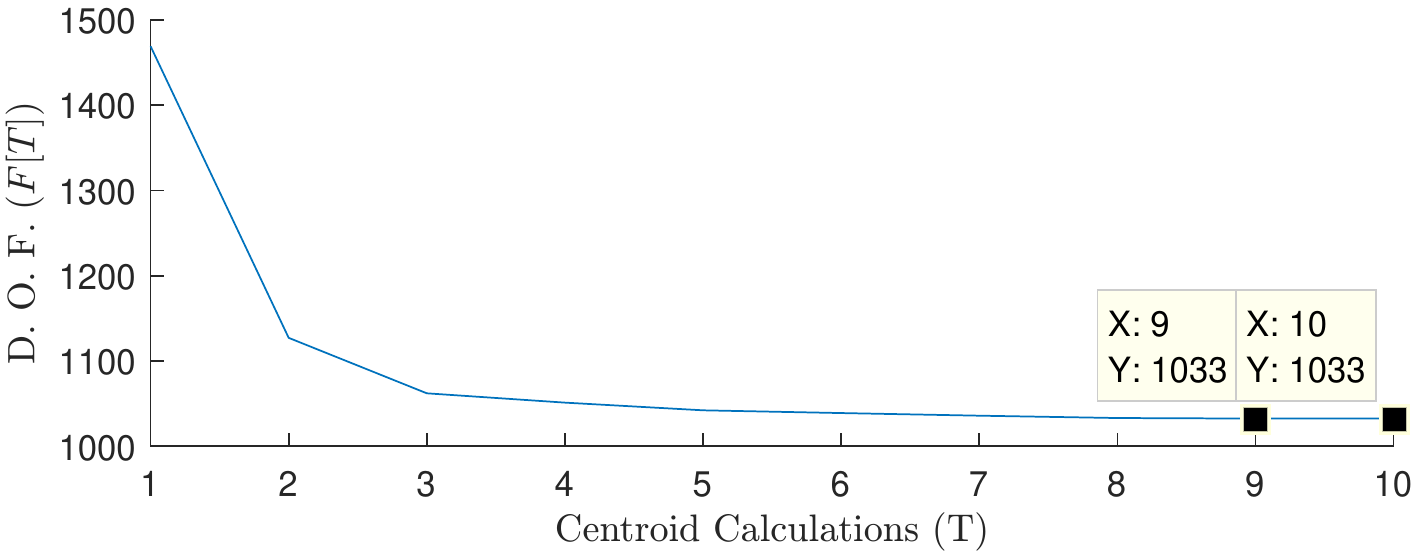}
    \caption{Evolution of Distance Objective Function $F[T]$ during execution of Algorithm~\ref{algorithm1} over a random directed network with $100$ nodes and $3$ clusters.} 
    \label{obj100nodes_50times50obj}
\end{figure}


\textbf{Evaluation over $1000$ Random Networks of $1000$ Nodes.}
We execute of Algorithm~\ref{algorithm1} over $1000$ random networks each consisting of $1000$ nodes, with diameter $D \in \{ 3, 4, 5 \}$. 
During the operation of Algorithm~\ref{algorithm1}, we aim to partition nodes into $k = 3$ clusters. 
For each of the $1000$ simulations, the nodes and the centroid positions are randomly selected in the region $[100, 100] \times [100, 100]$ with uniform probability. 
We present the distribution $\mathcal{F}[T]$ of the new centroid calculations $T$ until \eqref{stop_cond} holds for $100$ simulations of Algorithm~\ref{algorithm1}. 
Furthermore, we present the average value $\overline{F}[T]$ of the Distance Objective Function $F[T]$ in \eqref{DistanceObjectiveFunction}, averaged over the $1000$ simulations of Algorithm~\ref{algorithm1}. 

\begin{figure}[t]
    \centering
    \includegraphics[width=7.2cm]{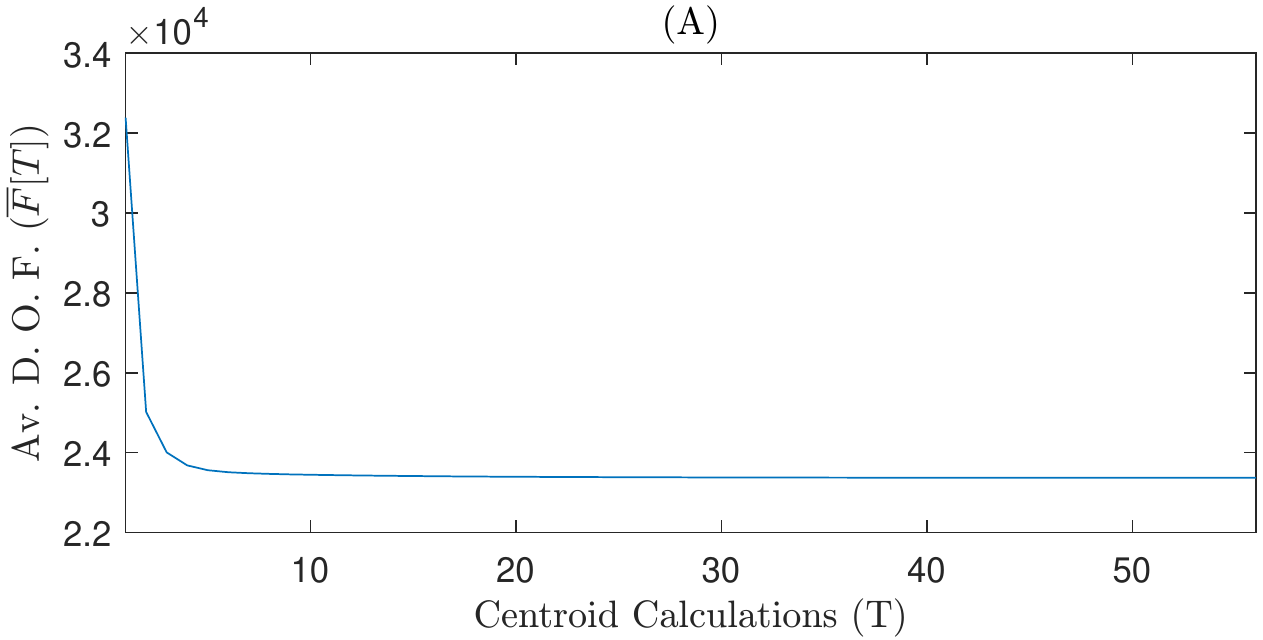}\\ \vspace{.3cm}
    \includegraphics[width=7.2cm]{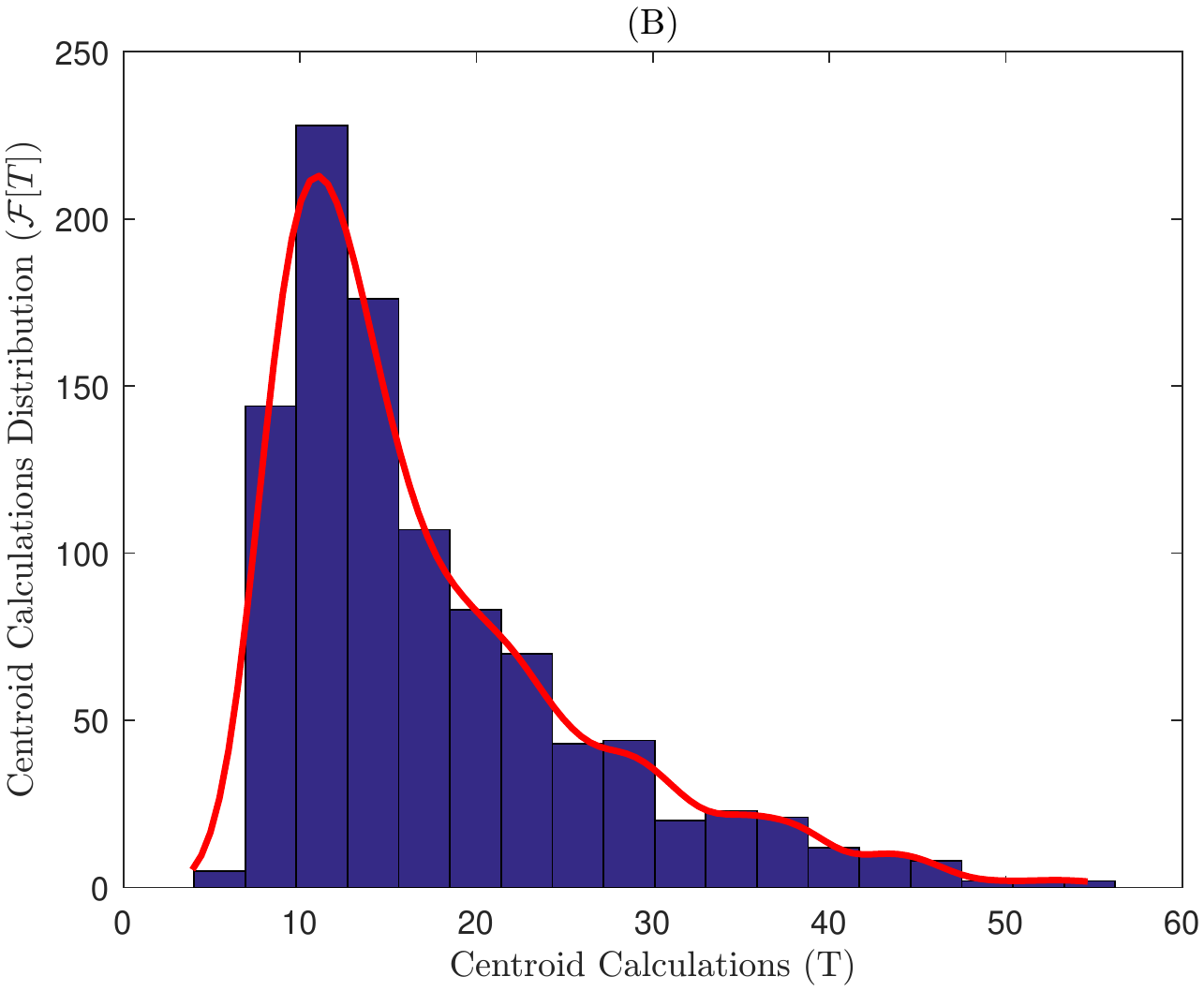}
    \caption{Executions of Algorithm~\ref{algorithm1} over $1000$ random directed networks of $1000$ nodes and $k = 3$ clusters. (A) Evolution of Average Distance Objective Function $\overline{F}[T]$ during execution of Algorithm~\ref{algorithm1}, averaged over $1000$ executions. (B) Distribution $\mathcal{F}[T]$ of number of new centroid calculations $T$ until \eqref{stop_cond} holds, for $1000$ executions.} 
    \label{1000nodes_1000times100_k3}
\end{figure}

In Fig.~\ref{1000nodes_1000times100_k3} (A), for $1000$ executions of Algorithm~\ref{algorithm1} we have that $\overline{F}[T]$ almost converges after $17$ centroid calculations $T$. 
Also, note that in Fig.~\ref{1000nodes_1000times100_k3} (A), $\overline{F}[T]$ is plotted for $T \in \{1, 2, ..., 56 \}$, where $56$ is the maximum value of $T$ for Algorithm~\ref{algorithm1} to converge over the $1000$ executions. 
In Fig.~\ref{1000nodes_1000times100_k3} (B), we have that the average value of $T$ for Algorithm~\ref{algorithm1} to converge over the $1000$ executions is $17.39$. 
The minimum value of $T$ is $5$, and the maximum is $56$ (also seen in Fig.~\ref{1000nodes_1000times100_k3} (A)). 
Furthermore, we can see that in most cases, the required $T$ for Algorithm~\ref{algorithm1} to converge over the $1000$ executions is in the set $T \in \{ 8, 9, ..., 20 \}$.

\textbf{Comparison with Previous Literature.} 
We now compare the performance of Algorithm~\ref{algorithm1} against algorithms \cite{2015:Oliva}, and \cite{2017:Zheng_Qin} in the current literature. 
We execute the three algorithms over $1000$ random networks each consisting of $1000$ nodes, with diameter $D \in \{ 3, 4, 5 \}$, and we aim to partition nodes into $k = 3, 6, 12$ clusters. 
Note that the main differences of Algorithm~\ref{algorithm1} compared to \cite{2015:Oliva}, and \cite{2017:Zheng_Qin} are mentioned in Section~\ref{compar_prev_work}. 
In \cite{2017:Zheng_Qin} we initially execute a $k$-means++ algorithm for the initial centroids. 
Furthermore, \cite{2015:Oliva} requires the underlying graph to be undirected. 
For this reason, during the operation of \cite{2015:Oliva} we make the randomly generated underlying digraphs undirected by enforcing that if $(v_l, v_j) \in \mathcal{E}$, then also $(v_j, v_l) \in \mathcal{E}$. 

In Table~\ref{tableavergecentr} we present the average number of centroid calculations $T$ over $1000$ executions of Algorithm~\ref{algorithm1}, \cite{2015:Oliva}, and \cite{2017:Zheng_Qin}. 
We can see that the performance of Algorithm~\ref{algorithm1} is close to the current literature but requires slightly more time steps. 
However, the aim of Algorithm~\ref{algorithm1} is to implement a communication efficient solution to the clustering problem. 
Specifically, Algorithm~\ref{algorithm1} operates with quantized values and requires less assumptions compared to the current literature i.e., in \cite{2017:Zheng_Qin} a $k$-means++ algorithm is initially executed and the network needs to be weight balanced, and in \cite{2015:Oliva} the network of each cluster is undirected and connected (see Section~\ref{compar_prev_work}). 

\begin{center}
\captionof{table}{Average Number of Centroid Calculations During Operation of Algorithm~\ref{algorithm1} (A), 
\cite{2015:Oliva} (B), 
\cite{2017:Zheng_Qin} (C), 
averaged for $1000$ executions over a random digraph of $1000$ nodes.}
\label{tableavergecentr}
\begin{tabular}{|c||r|r|r|}  
\hline
Algorithm & $k = 3$ & $k = 6$ & $k = 12$ \\
\cline{1-4}
(A) & $17.39$ & $20.79$ & $25.49$ \\
(B) & $16.55$ & $21.62$ & $24.16$ \\
(C) & $11.57$ & $16.86$ & $23.25$ \\
\hline
\end{tabular}
\end{center}

%
%
%
%

\section{Conclusions and Future Directions}\label{sec:conclusions}

In this paper, we have considered the problem of $k$-means clustering over a directed network. 
We presented a novel algorithm which is able to address the $k$-means clustering problem in a fully distributed fashion. 
We showed that our algorithm converges after a finite number of time steps, and we provided a deterministic upper bound on convergence time which relies on the network parameters. 
Finally, we demonstrated the operation of our proposed algorithm and compared its performance against other algorithms in the existing literature. 
Please note that to the best of the authors knowledge, this is the first work that tries to tackle the problem of distributed $k$-means clustering using quantized communication while also providing a thorough evaluation. 

Utilizing our algorithm's quantized nature in order to introduce privacy guarantees through cryptographic strategies, is our main future work direction.

\bibliographystyle{IEEEtran}
\bibliography{bibliography}

\end{document}